\documentclass[sigconf]{aamas}  
\usepackage{balance}  

\usepackage{booktabs}

\setcopyright{ifaamas}  
\acmDOI{}  
\acmISBN{}  
\acmConference[AAMAS'19]{Proc.\@ of the 18th International Conference on Autonomous Agents and Multiagent Systems (AAMAS 2019)}{May 13--17, 2019}{Montreal, Canada}{N.~Agmon, M.~E.~Taylor, E.~Elkind, M.~Veloso (eds.)}  
\acmYear{2019}  
\copyrightyear{2019}  
\acmPrice{}  

\usepackage{latexsym,graphicx,amssymb}
\usepackage{amsfonts,mathtools}
\usepackage{amsmath,enumerate}
\usepackage{float}
\usepackage{xspace}
\usepackage{paralist}
\usepackage{enumerate}
\usepackage{cases}
\usepackage{caption}
\usepackage{color}

\newenvironment{proofof}[1]{{\vspace*{5pt} \noindent\bf Proof of #1:  }}{\hfill\rule{2mm}{2mm}\vspace*{5pt}}

\newcommand{\alg}{\textsf{ALG}}
\newcommand{\opt}{\textsf{OPT}}

\begin{document}

\title{Well-behaved Online Load Balancing Against Strategic Jobs}
\titlenote{The authors are ordered alphabetically.}



\author{Bo Li}
\authornote{This work has been partially supported by NSF CAREER Award No. 1553385. Part of this work was done when Bo Li was visiting City University of Hong Kong.}
\affiliation{%
	\institution{Stony Brook University}
  	\city{Stony Brook, USA}
}
\email{boli2@cs.stonybrook.edu}

\author{Minming Li}
\authornote{City University of Hong Kong Shenzhen Research Institute, Shenzhen, P. R. China. The work described in this paper was partially sponsored by Project 11771365 supported by NSFC. The work described in this paper was supported by a grant from Research Grants Council of the Hong Kong Special Administrative Region, China (Project No. CityU 11200518)}
\affiliation{%
	\institution{City University of Hong Kong}
	\city{Hong Kong, China}
}
\email{minming.li@cityu.edu.hk}

\author{Xiaowei Wu}
\authornote{Part of work was done when the author was a postdoc at the City University of Hong Kong. The research leading to these results has received funding from the European Research Council under the European Community's Seventh Framework Programme (FP7/2007-2013) / ERC grant agreement No. 340506.}
\affiliation{%
	\institution{University of Vienna}
	\city{Vienna, Austria}
}
\email{wxw0711@gmail.com}

\begin{abstract}
	In the online load balancing problem on related machines, we have a set of jobs (with different sizes) arriving online, and we need to assign each job to a machine immediately upon its arrival, so as to minimize the makespan, i.e., the maximum completion time. 
	In classic mechanism design problems, we assume that the jobs are controlled by selfish agents, with the sizes being their private information. Each job (agent) aims at minimizing its own cost, which is its completion time plus the payment charged by the mechanism.
	Truthful mechanisms guaranteeing that every job minimizes its cost by reporting its true size have been well-studied [Aspnes et al. JACM 1997, Feldman et al. EC 2017].
	
	In this paper, we study truthful online load balancing mechanisms that are well-behaved [Epstein et al., MOR 2016].
	Well-behavior is important as it guarantees fairness between machines, and implies truthfulness in some cases when machines are controlled by selfish agents.
	Unfortunately, existing truthful online load balancing mechanisms are not well-behaved.
	We first show that to guarantee producing a well-behaved schedule, any online algorithm (even non-truthful) has a competitive ratio at least $\Omega(\sqrt{m})$, where $m$ is the number of machines.
	Then we propose a mechanism that guarantees truthfulness of the online jobs, and produces a schedule that is almost well-behaved.
	We show that our algorithm has a competitive ratio of $O(\log m)$.
	Moreover, for the case when the sizes of online jobs are bounded, the competitive ratio of our algorithm improves to $O(1)$. 	
	Interestingly, we show several cases for which our mechanism is actually truthful against selfish machines.
\end{abstract}



\maketitle


\section{Introduction} \label{sec:intro}

In the online load balancing problem, we have a group of $m$ machines $M$.
Each machine $i\in M$ has a speed $s_i$, which is the amount of work it processes in one unit of time.
There is a sequence of online jobs $J$, arriving in an arbitrary order.
Each job $j\in J$ has a size $p_j$, which is the amount of work it must be processed before completed.
In other words, the processing time of job $j$ on machine $i$ is given by $\frac{p_j}{s_i}$.
Upon the arrival of a job, the scheduling algorithm must decide immediately which machine the job goes to, with the objective of minimizing the makespan, i.e., the maximum completion time on any machine.
The problem is often referred to as the \emph{online load balancing problem}~\cite{jacm/AspnesAFPW97,jal/BermanCK00}.
In this paper, we consider the non-preemptive and non-migrate setting, i.e., once a job is scheduled on a machine, it must be processed continuously on this machine until it is finished.

The problem has been extensively studied for decades.
For identical machines (a special case when $s_i = 1$ for all $i\in M$), a $(2-\frac{1}{m})$-competitive algorithm was proposed by Graham~\cite{siamam/Graham69}, which is later improved to $1.9201$-competitive by Fleischer and Wahl~\cite{esa/FleischerW00}.
For related machines, the current best competitive ratio is $4.311$ for randomized algorithms and $5.828$ for deterministic algorithms~\cite{jal/BermanCK00}.
The problem in which jobs have arbitrary non-zero arrival times has also been studied~\cite{orl/ChenV97,tcs/NogaS01,approx/HuangKTWZ18}.

The game theory setting of this problem is also well-studied.
In this setting, each job is controlled by a selfish agent, whose objective is to finish its own job as early as possible. The agent is the only one who knows the size of the job, and for its own interest, it may misreport the size of the job.
In the mechanism design community, we are interested in load balancing algorithms that are \emph{truthful}, i.e., every job optimizes its own objective by telling the truth.
Existing truthful mechanisms often charge the jobs for being processed.
In this case, the objective of each job is to minimize its own total cost, which is the completion time of the job plus the payment.

The Greedy algorithm~\cite{jacm/AspnesAFPW97} is a classic truthful mechanism (with competitive ratio $O(\log m)$), which assigns each online job $j$ to the machine that finishes $j$ the earliest, with tie broken by machine identity.
Since the schedule minimizes the completion time of each online job, the mechanism is truthful.
Recently, an $O(1)$-competitive truthful mechanism is proposed by Feldman et al.~\cite{ec/FeldmanFR17}.
The main idea is to set dynamic prices on the machines such that the selfish behaviors of jobs result in a schedule similar to the one produced by Slow-Fit~\cite{jacm/AspnesAFPW97}, an $O(1)$-competitive (non-truthful) algorithm for the online load balancing problem on related machines.

Another well-studied game theory model for load balancing assumes that the jobs are given offline, while the machines are controlled by selfish agents, with the speed being its only private information~\cite{siamcomp/ChristodoulouK13,mor/EpsteinLS16}.
To guarantee truthfulness of machines, a scheduling mechanism pays the machines for working~\cite{stacs/AndelmanAS05,esa/Kovacs05} (otherwise all machines will misreport a speed of $0$ to get rid of any assignment).
Each machine aims at maximizing the payment it receives, minus its own makespan, which is the time by which all jobs assigned to it are finished.
It has been shown that as long as the scheduling algorithm is \emph{monotone}, there exists a payment function such that the mechanism is truthful~\cite{stacs/AndelmanAS05}.
Roughly speaking, monotonicity requires that, if a machine claims a higher speed, then its workload (or makespan) does not decrease.

A natural and interesting question is, if both the machines and jobs are controlled by selfish agents, is it possible to design a truthful mechanism with bounded competitive ratio?

In this paper, we make a few steps towards solving this question.
We first show some simple mechanisms that are truthful for both the jobs and machines, but have large competitive ratios, e.g., $O(m)$.
In order to improve the competitive ratio, we relax the truthfulness of machines to {\em well-behavior} of the schedule, which guarantees a different kind of monotonicity with respect to speeds of machines.

\subsection{Well-behaved Schedules}

The concept of ``well-behaved schedule'' was first studied by Epstein et al.~\cite{mor/EpsteinLS16} for the case when related machines are controlled by selfish agents.
A well-behaved schedule guarantees that a machine has workload no smaller than that of any slower machine.
Note that well-behavior is closely related to truthfulness.
Imagine two machines with speeds $s_1 < s_2$.
If the schedule is not well-behaved, i.e., the faster machine has workload less than that of the slower machine, then the faster machine has incentive to lie.
More specifically, if the faster machine claims speed $s^2_1/s_2$, then its workload becomes larger: the identities of the two machines are swapped.
Therefore the mechanism is not monotone, and not truthful.

Without knowing the future jobs to arrive, it is very difficult (or maybe impossible) to achieve a monotone schedule with small competitive ratio.
Since monotonicity is defined on the set of online jobs, which the algorithm is oblivious of, the natural candidate algorithm must somehow adopt independent assignment of jobs, i.e., the assignment of any two jobs are independent.
However, to guarantee a bounded competitive ratio, the assignments of jobs must be correlated, e.g., if a machine receives many jobs, then the next one should avoid going to this machine.

On the other hand, well-behavior is a more approachable measurement of the schedule.
Indeed, since a well-behaved schedule guarantees that no matter what jobs may arrive, the faster machines receive more jobs than slower machines, without knowing what jobs may arrival, the risk-averse machines have no incentive to lie.

Unfortunately, existing truthful mechanisms~\cite{jacm/AspnesAFPW97,ec/FeldmanFR17} against strategic jobs are not well-behaved.
As a simple example, let us consider the Greedy algorithm~\cite{jacm/AspnesAFPW97}.
Imagine that we have two machines with speed $1$ and $1+\epsilon$, respectively, where $\epsilon>0$ is arbitrarily small.
Let there be two jobs, the first one with size $1$ and the second one with size $\frac{1}{\epsilon}$.
Then the slower machine has workload $\frac{1}{\epsilon}$, while the faster machine has workload $1$, much smaller than the workload of the slower machine.
Actually, as we will show later, being well-behaved is difficult: any well-behaved algorithm has competitive ratio $\Omega(\sqrt{m})$.

\subsection{Our Results}\label{ssec:our_results}

In this paper, we consider a stronger definition of well-behaved schedules. 
We call a schedule {\em well-behaved} if the makespan (the total workload divided by the speed) of a machine is no smaller than that of any slower machines.
Note that under this definition, a faster machine gets more workload than proportional to its speed.
As a consequence, it is more likely for the machines to be truthful, given that the payment function are defined as in~\cite{stacs/AndelmanAS05}.

First, we show that being well-behaved is difficult, even under the weaker version~\cite{mor/EpsteinLS16}.
We show that there exists a sequence of online jobs on a group of machines such that any well-behaved schedule performs badly.

\begin{theorem}[Hardness of Well-Behaved Schedule] \label{th:hardness}
	Well-behaved algorithms have competitive ratios $\Omega(\sqrt{m})$.
\end{theorem}

On the other hand, we show that if we relax the requirement of well-behavior slightly, i.e., we require a machine $i$ to have a no smaller makespan than machine $i'$ only if $s_i \geq 2 s_{i'}$, then a much smaller competitive ratio can be obtained.
We propose such kind of ``almost well-behaved'' scheduling algorithm that is truthful for online jobs, and achieves a competitive ratio of $O(\log m)$. Moreover, we show that the competitive ratio improves to $O(1)$ if the online jobs have similar sizes. 

\begin{theorem} \label{th:well-behaved}
	There exists an $O(\log m)$-competitive algorithm that is truthful for jobs, and produces an almost well-behaved schedule for any online jobs.
	Moreover, the competitive ratio improves to $O(\frac{p_\text{max}}{p_\text{min}})$ when $p_j\in [p_\text{min}, p_\text{max}]$ for all jobs $j\in J$.
\end{theorem}

Besides being almost well-behaved, our mechanism (detailed description in Section~\ref{sec:mechanism}) has several other desirable properties.

First, our mechanism is \emph{fair} in the sense that if two machines have the same speed, then their workload differs by at most one job. 
Our mechanism guarantees that a job is assigned to a machine only when the machine has the minimum workload among all machines with the same speed. 
Indeed, since we break tie randomly, the machines with the same speed have the same {\em expected} makespan. 
Such fairness is essentially the same with the notion of envy-free up to any item in resource allocation literature
when all agents have the same entitlement~\cite{ec/caragiannis2016,soda/plaut2018}.
Second, our mechanism is \emph{anonymous}. 
In our mechanism, we guarantee that the schedule produced is determined only by the speeds submitted by the machines, which means that if a machine changes its identity, the resulting schedule remains unchanged.
As a consequence, our mechanism is also {\em symmetric}, i.e., if two machines swap their speeds, then the resulting workloads will also be swapped.
We note that anonymity and symmetry are also important properties 
in the cooperative game theory and facility location literature~\cite{book/nisan2007}, 
which keeps the system stable,
and are necessary conditions for Shapley value~\cite{book/shapley1953}
or Nash's bargaining solution~\cite{Econometrica/nash1950,book/moulin1991}.

\paragraph{Our Technique.}
To guarantee that the makespan of a faster machine is not smaller than that of a slower machine, we assign each online job to a machine only if the well-behavior is maintained.
In other words, we maintain the invariant that the makespans of machines are non-decreasing with respect to the speeds, and we assign job $j$ to machine $i$ only if after the assignment, the makespan of machine $i$ is not larger than that of any faster machine.

Two questions arise. First, how can we make the mechanism truthful for online jobs, given that sometimes we assign a job to a machine other than the one that finishes the job the earliest? Second, is the competitive ratio of the algorithm bounded?

To answer the first question, we use the technique of dynamic posted pricing~\cite{ec/HartlineR09,ec/FeldmanFR17}. We show that depending on the current makespans of machines, we can set prices on the machines (which is the charging to the jobs) such that the selfish behaviors of online jobs result in a well-behaved schedule.
As a simple example, suppose there are two machines with speeds $s_1 = 1$ and $s_2 = 2$, and the current makespan of machine $1$ is $1$ while machine $2$ has makespan $2$. We set a price $0.5$ on machine $1$ and price $0$ on machine $2$.
Then if the next job $j$ has size $p_j \leq 1$, its cost when going to machine $1$ is $1 + p_j + 0.5$, which is no larger than $2 + {p_j}/2$, the cost when going to machine $2$.
On the other hand, if $p_j > 1$, then going to machine $2$ has a strictly smaller cost.
Hence the next job goes to machine $1$ only if its size is at most $1$, which implies that the resulting schedule remains well-behaved.

For the second question, in contrast to the $\Omega(\sqrt{m})$ hardness result, we show that if we only impose the monotonicity on the makespans of machines whose speeds differ by a factor of $2$, then a much better competitive ratio $O(\log m)$ can be obtained.
The main idea is similar to the analysis of the Greedy algorithm~\cite{jacm/AspnesAFPW97}.
We show that the makespans of machines with speeds $s$ and $\frac{s}{2}$ differ by at most $O(\opt)$, as otherwise (by the way we set the dynamic prices) the jobs will stop going to the faster machines.
Combining this idea with the fact that machines with speed less than $\frac{s}{m}$, where $s$ is the maximum speed of a machine, has little contribution to the whole schedule, we can show that the makespan of our schedule is bounded by $O(\log m)\cdot \opt$.

Interestingly, as illustrating examples, we show that our mechanism is indeed truthful in some cases.

\begin{theorem} \label{th:truthful}
	Our mechanism is also truthful for machines when $m=2$, or when online jobs have unit size.
\end{theorem}

As far as we know, our work is the first result that achieves truthfulness for an online load balancing problem in which both the jobs and machines are controlled by selfish agents.
A natural and fascinating open question is whether there exist non-trivial mechanisms that are truthful for both jobs and machines in general, and achieve a competitive ratio $o(m)$.

\subsection{Other Related Work}

\paragraph{Truthful Mechanism for Selfish Machines.}
The work of scheduling mechanisms against selfish machines was initiated by Nisan and Ronen~\cite{stoc/NisanR99}, in which the case of unrelated machines is studied.
They propose a truthful mechanism (which adopts independent assignments of jobs) that is $m$-competitive, and show that no truthful mechanism can do better than $2$-competitive.
The hardness result has been improved to $1+\sqrt{2}\approx2.414$ for $3$ or more machines~\cite{algorithmica/ChristodoulouKV09}, and to $1+\phi\approx 2.618$ when the number of machines tends to $\infty$~\cite{algorithmica/KoutsoupiasV13}.

On the other hand, the problem on related machines is much easier.
Constant competitive truthful mechanisms are proposed in~\cite{stacs/AndelmanAS05,esa/Kovacs05}, and the state-of-the-art mechanisms are the deterministic PTAS by Christodoulou and Kov{\'{a}}cs~\cite{siamcomp/ChristodoulouK13}, and by Epstein et al.~\cite{mor/EpsteinLS16}.

\paragraph{Truthful Mechanism for Selfish Online Jobs.}
The Greedy algorithm is the most well-known online load balancing mechanism that is truthful against selfish online jobs.
The algorithm has been shown to be $2$-competitive on identical machines~\cite{siamam/Graham69}; $\Theta(\log m)$-competitive on related machines, and $m$-competitive on unrelated machines~\cite{jacm/AspnesAFPW97}.
Recently, a constant competitive truthful mechanism that adopts dynamic posted prices on machines is proposed~\cite{ec/FeldmanFR17}.

\medskip

There are other works~\cite{tpds/DutotPRT11,spaa/SkowronR13} focusing on producing fair schedules in a cooperative game-theoretic perspective.
However, since they do not study the truthfulness of jobs or machines, the model we study in this paper is different from theirs.

\paragraph{Hardness of Online Load Balancing.}
In contrast to the PTAS on related machines obtained when all jobs are given offline, due to the nature of unknown future, it is impossible to achieve any competitive ratio strictly smaller than $1.5$, even on two identical machines, and when both machines and all jobs are truthful: imagine $2$ jobs of size $1$, followed by another job with size $2$.
On related machines, it is shown in~\cite{jal/BermanCK00} that no deterministic algorithm can achieve a competitive ratio smaller than $2.438$, and no randomized algorithm can do better than $1.8372$-competitive.

\section{Preliminaries}

In this paper, we use $J$ to denote the set of online jobs to be scheduled on $m$ related machines $M$.
Let $p_j$ be the size of job $j\in J$, and $s_i$ be the speed of machine $i\in M$.
The size $p_j$ is the only private information of job $j\in J$.
Without loss of generality (by scaling) we assume that the minimum speed of a machine and the smallest size of an online job are both $1$.
Let $[k]$ denote $\{1,2,\ldots,k\}$, for any positive integer $k$.
We use $\log$ to denote $\log_2$.

We assume all jobs arrive one-by-one at time $0$.
There is no deadline or weight on the jobs.
Each job must be assigned to some machine immediately upon its arrival.
The objective is to assign jobs to machines to minimize the makespan.
In the following, we use \emph{workload} on a machine to denote the total size of jobs assigned to the machine; \emph{makespan} on the machine to denote its workload divided by its speed.

\paragraph{Selfish Agents.}
In our problem, we assume that the jobs are controlled by independent selfish agents, who is the only person that knows the size of the job.
The objective of a selfish job is to minimize its own completion time plus payments to the mechanism.
For the case when machines are also controlled by selfish agents, we assume that only the agent knows the speed of the machine, and its objective is to maximize (payment - makespan), where the payment is determined by the mechanism based on the final schedule.

\paragraph{Mechanism.}
We consider mechanisms that operate in two phases.
In phase-$1$, machines submit their speeds to the mechanism before any job arrives.
In the case when machines are selfish, the reported speed might be different from the true value.
The mechanism then announces the speeds of machines (which might not be equal to the speeds submitted).
In phase-$2$, each online job arrives and submits a size to the mechanism.
Each online job is informed of the announced speeds of machines, and the current schedule.
Based on the size of the job, the mechanism decides immediately which machine the job goes to, together with a charging to the job.
After assigning all jobs, the mechanism decides a payment to every machine.

\medskip

We call a mechanism \emph{truthful} if and only if every job maximizes its own utility by telling the truth, i.e., their true sizes.
We call a mechanism \emph{well-behaved} if the makespan of a machine is no smaller than that of any slower machine.
Throughout this paper, we use $\alg$ and $\opt$ to represent the makespan of our mechanism and the offline optimal makespan, respectively.

In the remaining part of this paper, we assume that before the machines submit their speeds, we fix a random ordering of machines that is used to break tie (when multiple machines with the same speed have the same makespan).

\subsection{Monotone Schedule}

It is shown in~\cite{stacs/AndelmanAS05} that for the case when machines are controlled by selfish agents, the scheduling mechanism is truthful if and only if the schedule is monotone.

\begin{definition}[Monotone]
	Consider any machine $i\in M$, and fix the speeds of all other machines arbitrarily.
	A scheduling mechanism is monotone if the workload of machine $i$ is non-decreasing w.r.t. the speed of machine $i$.
\end{definition}

More specifically, let $L(b)$ be the workload of the machine when the claimed speed is $\frac{1}{b}$, i.e., processing one unit of job requires $b$ units of time.
Let
\begin{equation}
\textstyle P(b) = b\cdot L(b) + \int_{b}^{\infty} L(t) dt \label{eq:payment}
\end{equation}
be the payment function to the machine.
It can be shown that if $L(b)$ is non-decreasing, then the mechanism (with payment function~\eqref{eq:payment}) is truthful~\cite{stacs/AndelmanAS05}.
For completeness, we give a short proof here.
\begin{lemma}
	The mechanism is truthful if $L(b)$ is non-increasing.
\end{lemma}
\begin{proof}
	Let $s$ be the true speed of the machine and $b^* = \frac{1}{s}$.
	It suffices to show that $b^*$ maximizes the utility function $P(b)-b^*\cdot L(b)$.
	We show that $\forall b\neq b^*$, we have $P(b)-b^*\cdot L(b) \leq P(b^*)-b^*\cdot L(b^*)$.
	
	If $b < b^*$ (claim to be faster than truth), then we have
	\begin{align*}
	& \textstyle P(b)-b^*\cdot L(b) = \int_b^\infty L(t)dt - (b^*-b)\cdot L(b) \\
	= & \textstyle \int_{b^*}^\infty L(t) dt + \int_{b}^{b^*} \left( L(t)-L(b) \right) dt
	\leq  P(b^*)-b^*\cdot L(b^*),
	\end{align*}
	where the inequality comes from the fact that $\forall t\geq b$, $L(t)\leq L(b)$.
	
	Similarly for $b > b^*$ (claim to be slower than truth), then
	\begin{align*}
	& \textstyle P(b)-b^*\cdot L(b) = \int_b^\infty L(t)dt + (b-b^*)\cdot L(b) \\
	= & \textstyle \int_{b^*}^\infty L(t) dt + \int_{b^*}^{b} \left( L(b)-L(t) \right) dt
	\leq P(b^*)-b^*\cdot L(b^*),
	\end{align*}
	where the inequality comes from $L(t)\geq L(b)$ for all $t\leq b$.
\end{proof}

\section{Na\"{i}ve Algorithms} \label{sec:naive}

In this section, we present two simple scheduling mechanisms that are truthful for both the jobs and the machines.
Basically, these mechanisms are two extreme cases, with one allocating all jobs on a single machine
and the other allocating all jobs uniformly.
As we will see, the competitive ratios of the mechanisms are very large.

We first propose a simple VCG-style mechanism ${\mathcal M}_\text{VCG}$ that employs payments to machines but no charging to jobs.
The assignment of jobs is trivial: we assign all online jobs to the machine with the highest claimed speed.

\paragraph{Mechanism ${\mathcal M}_\text{VCG}$.}
Let $t_1\leq \ldots \leq t_m$ be the speeds claimed by machines, where we break tie by the pre-fixed ordering.
We assign all jobs to machine $m$, and set the payment to the machine be $\frac{1}{t_{m-1}}\sum_{j\in J}q_j$, where $q_j$ is the processing time claimed by job $j$.
For truth-telling machines and jobs, we have $t_i = s_i$ for all $i\in M$ and $q_j = p_j$ for all $j\in J$.

\paragraph{Truthfulness of ${\mathcal M}_\text{VCG}$.}
Observe that the schedule does not depend on the sizes of online jobs, and thus is truthful for jobs (misreporting the size is not beneficial).
The mechanism is also truthful for the machines, as all jobs are processed as a grand bundle, with size $\sum_{j\in J} p_j$.
A machine is willing to process the bundle only if its actual processing time on the bundle is less than the payment.
Since the payment is determined by the second highest speed, the mechanism (which is actually a second price auction) is truthful for machines (bidders).

\paragraph{Competitive Ratio of ${\mathcal M}_\text{VCG}$.}

Since all jobs are assigned to machine $m$ with the highest speed $s_m$, $\alg$ is the makespan of machine $m$.
On the other hand, since $s_m\geq s_i$ for all $i\in [m]$, the processing time of any job on machine $m$ is the smallest among all machines.
Hence in the optimal schedule, the summation of makespans of all machines is at least $\alg$.
Therefore we have $\alg\leq m\cdot \opt$.

\medskip

Another simple truthful mechanism ${\mathcal M}_\text{Greedy}$ works as follows.
Regardless of the speeds submitted by the machines, the mechanism announces $s_i = 1$ for all machine $i$, and assigns each online job to the machine that finishes the job the earliest. In other words, the mechanism runs the Greedy algorithm on identical machines.

\paragraph{Truthfulness of ${\mathcal M}_\text{Greedy}$.}
Obviously, the mechanism is truthful for online jobs, as the completion time of every job is minimized, and there is no charging to the jobs.
The mechanism is also truthful for the machines, as the speeds submitted by machines are irrelevant to the schedule.

\paragraph{Competitive Ratio of ${\mathcal M}_\text{Greedy}$.}
Consider the optimal schedule (with makespan $\opt$).
Imagine that we keep the same schedule but change the speeds of all machines to $1$.
Let $\opt'$ be the resulting makespan, $s_\text{min} = \arg\min_{1\leq i\leq m}\{s_{i}\}$ and $s_\text{max} = \arg\max_{1\leq i\leq m}\{s_{i}\}$.
Then we have $\opt' \leq \frac{s_\text{max}}{s_\text{min}}\cdot \opt$, as the makespan of every machine is increased by a factor at most $\frac{s_\text{max}}{s_\text{min}}$.
Since Greedy is $2$-competitive for load balancing on identical machines, we have $\alg \leq 2\cdot \opt' \leq \frac{2 s_\text{max}}{s_\text{min}}\cdot \opt$.

As it appears to be, since we totally ignore the speeds of machines, the competitive ratio is determined by how different the speeds of machines are.

\section{A Well-Behaved Mechanism} \label{sec:mechanism}

As we have shown in the previous section, to achieve complete truthfulness, we may have to (somehow) ignore the speeds of the machines, or the sizes of online jobs. However, such kind of algorithms usually suffer from huge competitive ratios.

In this section, we focus on a more approachable measurement of schedules, the well-behavior.
As we have shown in Section~\ref{sec:intro}, since the machines are oblivious of the online jobs, a well-behaved mechanism (somehow) guarantees that the risk-averse rational selfish machines have no incentive to lie.

\subsection{Hardness of Well-behaved Schedules}

We first prove in this section that, while it seems natural and easy, it is actually impossible for a well-behaved scheduling algorithm to achieve an $o(\sqrt{m})$ competitive ratio in the online setting.
Moreover, our $\Omega(\sqrt{m})$ hardness result holds even when jobs and machines are truthful, and for the weaker definition of well-behavior as given in~\cite{mor/EpsteinLS16} (faster machines have no smaller workload).

\begin{proofof}{Theorem~\ref{th:hardness}}
	Consider a set of machines with speeds $s_i = 1+2^{i-m}$, where $i= 1,2,\ldots,m$.
	Let there be $m$ online jobs with different sizes $s_1,\ldots,s_m$.
	Obviously, the optimal schedule completes all jobs before time $1$ by assigning the job of size $s_i$ to machine $i$.
	
	However, we show that if the jobs arrive in increasing order of sizes, then any well-behaved schedule has makespan at least $\Omega(\sqrt{m})$, which yields a $\Omega(\sqrt{m})$ hardness result on the competitive ratio on any well-behaved scheduling algorithm.
	
	We prove by induction on the jobs that the number of jobs on the machines that receive at least one job, is strictly increasing w.r.t. the speed.
	Note that the statement implies that at the end of the algorithm, the number of jobs $k$ on the fastest machine $m$ is at least $\Omega(\sqrt{m})$, as otherwise the total number of jobs assigned is at most $\frac{k(k-1)}{2}<m$.	
	The statement trivially holds when there is no job, and after the first job arrives.
	
	Suppose the statement holds true before job $j$ arrives.
	
	We show that the statement holds true after assigning job $j$, given that the previous schedule is well-behaved.
	Suppose job $j$ is assigned to machine $i$. It suffices to show that after the assignment, the number of jobs on machine $i$ is smaller than that of machine $i+1$.
	Assume the contrary and let $c$ be the number of jobs on both machines. 
	Then the total workload of machine $i$ is at least $c+\frac{2^j}{2^m}$, while the workload on machine $i+1$ is upper bounded by
	\begin{equation*}
	\sum_{l=1}^c\left( 1+\frac{2^{j-l}}{2^m} \right) < c+\frac{2^j}{2^m}.
	\end{equation*}
	In other words, the workload (and also the makespan) of machine $i+1$ is strictly less than that of machine $i$, which contradicts the definition of well-behaved schedule. 
\end{proofof}

\subsection{Rounded Speeds}

Observe that well-behavior sometime is not necessary to achieve a truthful mechanism.
For example, recall the mechanism ${\mathcal M}_\text{Greedy}$ proposed in Section~\ref{sec:naive}.
Since the final schedule is irrelevant to the speeds, well-behavior is not needed to achieve truthfulness.
However, such kind of "speed insensitive" schedules often suffer from bad competitive ratios.
Nevertheless, this idea inspires us to develop a schedule that is almost well-behaved, with a limited speed insensitivity.
By relaxing the requirements slightly, we show that we can dramatically improve the competitive ratio, 
from $\Omega(\sqrt{m})$ to $O(\log m)$ for the general case
and to $O(1)$ for the case when the sizes of online jobs are bounded.

Recall that mechanism ${\mathcal M}_\text{Greedy}$
has a competitive ratio proportional to the difference of machine speeds.
The main idea behind the design of our new algorithm is, instead of ignoring the speeds completely, to achieve a $o(\sqrt{m})$ competitive ratio, we only regard a set of machines identical if their speeds are similar, e.g., differ by a multiplicative factor of $2$.
This idea motivates the design of the following mechanism ${\mathcal M}_\text{PPR}$, 
where $PPR$ is short for ``posted prices for rounded-down speeds''.

\paragraph{Mechanism ${\mathcal M}_\text{PPR}$.}
For every speed $s_i$ submitted by machine $i$, we round $s_i$ down to the largest power of $2$.
In other words, if $s_i \in [ 2^k, 2^{k+1} )$, then the mechanism announces $s_i = 2^k$.
We order the machines in non-decreasing order of rounded speeds: $s_1\leq s_2\leq \ldots \leq s_m$, where we break tie by the pre-fixed ordering.
At any moment, depending on the current makespans of machines, we set a price $\rho_i$ on every machine $i$ (that is independent of the incoming job).
When a new job arrives, it goes to the machine that minimizes its own cost, i.e., behaves selfishly.
At the end, let the payment to machines be defined as in Equation~\eqref{eq:payment}.

\medskip

Although we have not specified how we set the prices on the machines, we know that as long as the posted prices do not depend on the size of the new job, the mechanism is truthful for jobs (misreporting the size is not beneficial).
From now on, we assume that every machine has speed equals to some power of $2$.

In the following, we show that
\begin{enumerate}
	\item there is a way for mechanism ${\mathcal M}_\text{PPR}$ to set prices on machines such that the resulting schedule is well-behaved: machine with higher speed has no smaller makespan;
	
	\item the scheduling algorithm is $O(\log m)$-competitive.
\end{enumerate}

For ease of discussion, we give some necessary definitions first.

\begin{definition}
	For any job $j$, let $C_i(j)$ be the makespan of machine $i$ before job $j$ arrives; let $C_i$ be the makespan of machine $i$ at the end.
	Similarly, let $\rho_i(j)$ be the price on machine $i$ before job $j$ arrives.
\end{definition}

By definition, each online job $j$ goes to machine $i$ that minimizes the cost $C_i(j)+\frac{p_j}{s_i}+\rho_i(j)$.
We assume that the jobs are indexed by $1,2,\ldots,|J|$.
Thus $C_i(j+1)$ denotes the makespan of machine $i$ after assigning job $j$.

As mechanism ${\mathcal M}_\text{PPR}$ does not know when the online sequence of jobs stops, we maintain a well-behaved schedule at all time.
In other words, for any job $j$ and for any two machines $i$ and $i'$ with $s_i > s_{i'}$, we have $C_i(j) \geq C_{i'}(j)$.

\subsection{Dynamic Prices on Machines}

In this section, we show how mechanism ${\mathcal M}_\text{PPR}$ sets dynamic prices on machines such that (1) the prices are independent of the new job; (2) the resulting schedule is well-behaved.

Consider the moment before job $j$ arrives.
First, if there are multiple machines with the same speed, then except for the one with minimum makespan (break tie by the pre-fixed ordering), we set the prices on other machines to be infinity.
In other words, we never assign the next job to these machines.
For convenience of discussion, we assume w.l.o.g. that the speeds of all machines are different when the next job arrives.

Let $\pi_m = 0$.
For each machine $i<m$, define
\begin{equation*}
\textstyle \pi_i := \frac{s_i}{s_{i+1}}\cdot \big(C_{i+1}(j) - C_i(j) \big).
\end{equation*}

We set the price on machine $i$ to be $\rho_i(j) = \sum_{l\geq i}\pi_l$.

Note that by definition, the price $\rho_i(j)$
depends only on the recurrent makespans $C_i(j)$ of machines,
which is irrelevant to the size of the new job.
In our mechanism, the price of going to machine $m$ is $0$, 
while the price of going to machine $i$, where $i<m$, is $\pi_i$ more expensive than going to machine $i+1$.
Intuitively, the prices on the machines incentivize the jobs to go to faster machines, unless the makespan of a slower machine is much smaller.

\paragraph{Example.}
Suppose there are $m=3$ machines, which speed $1$, $2$ and $4$.
Let the online jobs have size $6$, $4$, $1$, $0.6$ (ordered by the arrival order).
By definition initially all prices on the machines are $0$.
Thus the first job goes to the third machine (with speed $4$), resulting in a makespan of $1.5$.
Before the second job arrives, the prices on the machines are $0.75$, $0.75$ and $0$.
Thus the second job (with size $4$) still goes to the third machine, as the objective is $2.5$, strictly smaller than $4$ and $2.75$, the objective when it goes to the first and second machine, respectively. 
Then the prices are updated to $1.25$, $1.25$ and $0$, and thus the third job (with size $1$) goes to second machine.
After that, the prices become $1.25$, $1$ and $0$, and the last job (with size $0.6$) goes to the second machine.
It is easy to see that the schedule the algorithm produces at every step is well-behaved.

\begin{lemma} \label{lemma:well-behave}
	The above dynamic prices on machines result in a schedule that is well-behaved, if jobs behave selfishly.
\end{lemma}
\begin{proof}
	We prove the lemma by induction on the jobs: we show that as long as the schedule is well-behaved, then after assigning the new job, the schedule remains well-behaved.
	
	By definition, the schedule is well-behaved when there is no job.
	Now suppose that the schedule is well-behaved before job $j$ arrives.
	We show that after assigning job $j$ to the machine with minimum cost, the schedule remains well-behaved.
	
	Suppose job $j$ (of size $p_j$) goes to machine $i$.
	It suffices to show that after the assignment, the makespan $C_i(j+1)$ of machine $i$ is no larger than that of machine $i+1$, which is $C_{i+1}(j)$, as this is the only place where well-behavior may cease to hold.
	
	Given that $i$ minimizes the cost of job $j$, we have
	\begin{equation*}
	\textstyle C_i(j) + \frac{p_j}{s_i} + \sum_{l\geq i}\pi_l \leq C_{i+1}(j) + \frac{p_j}{s_{i+1}} + \sum_{l\geq i+1}\pi_l,
	\end{equation*}
	which is equivalent to
	\begin{align*}
	\textstyle \left(\frac{1}{s_i}-\frac{1}{s_{i+1}}\right)\cdot p_j &\leq C_{i+1}(j)-C_i(j) - \pi_i \\
	& \textstyle  = s_i\cdot \left( \frac{1}{s_i}-\frac{1}{s_{i+1}} \right)\cdot (C_{i+1}(j) - C_i(j)),
	\end{align*}	
	where the equality holds by definition of $\pi_i$.
	
	Rearranging the inequality, we have $C_i(j) + \frac{p_j}{s_i} \leq C_{i+1}(j)$.
	In other words, after assigning the new job $j$ to machine $i$, the makespan of $i$ is still no larger than that of $i+1$.
	Hence the schedule remains well-behaved after assigning the new job.
\end{proof}

Note that we can only prove that job $j$ goes to a machine without violating the well-behavior of the schedule.
As there are possibly many candidate machines job $j$ can go to without violating this property, the actual assignment depends on the makespans of all machines.
More specifically, mechanism ${\mathcal M}_\text{PPR}$ cannot guarantee that the new job will be assigned to the fastest or slowest candidate machine (which makes the analysis more difficult).

\subsection{Competitive Ratio}

In this section, we bound the competitive ratio of mechanism ${\mathcal M}_\text{PPR}$ by $O(\log m)$.
As we will show later in Lemma~\ref{lemma:O(1)-comp}, the competitive ratio can also be upper bounded by $O(p_\text{max})$, where $p_\text{max}$ is the maximum size of an online job (recall that the minimum size of online jobs is scaled to $1$).
Thus we can upper bound the competitive ratio by $\min\{ \log m, p_\text{max} \}$.

Consider the final schedule produced by mechanism ${\mathcal M}_\text{PPR}$, after assigning all jobs.
Recall that $\opt$ is the optimal makespan and $\alg$ is the makespan of mechanism ${\mathcal M}_\text{PPR}$.
Let $\opt_2$ be the optimal makespan when the speeds of all machines are rounded down to powers of $2$.
Then immediately we have $\opt \leq \opt_2 \leq 2\cdot \opt$, as the optimal schedule before the rounding gives a feasible schedule with makespan at most $2\cdot \opt$ after the rounding.
Thus it suffices to show that $\alg = O(\log m)\cdot \opt_2$.
In the following, we assume that all machines have their speeds rounded down to powers of $2$.

For convenience of discussion, we order the machines such that $C_1 \leq C_2\leq \ldots \leq C_m$.
Since the schedule is well-behaved, we also have $s_1\leq s_2\leq \ldots\leq s_m$.
Then we have $\alg = C_m$. Let the speed of the fastest machine be $s_m = 2^K$.

\paragraph{Proof Outline.}
We first show that the makespans of any two machines with speed $2^k$ and $2^{k+1}$ differ by $O(\opt)$ at any moment (Lemma~\ref{lemma:difference-of-C}).
Intuitively, if the makespan of the faster machine is $\omega(\opt)$ larger than that of the slower machine, then the price our mechanism sets on the faster machine is $\omega(\opt)$ higher than the price on the slower machine.
Then we show that some job assigned to the faster machine has a strictly smaller objective when assigned to the slower machine, which is a contradiction.
Equipped with this property, it suffices to show that only $O(\log n)$ difference speeds matters: if a machine is $2^{O(\log n)} > n$ times slower than the fastest machine, then its contribution to the total workload is negligible.

\medskip

Let $R_k = \{ i\in[m] : s_i = 2^k \}$ be the set of machines with speed $2^k$, and $R_{\geq k} = \bigcup_{j\geq k} R_j$ be the machines with speed at least $2^k$.

We prove the following lemma, which is crucial for the proof of the competitive ratio.

\begin{lemma} \label{lemma:difference-of-C}
	Suppose $C_i\geq C$ for all machine $i\in R_{\geq k}$, for some $C > 0$, then for all machine $l \in R_{\geq k-1}$, we have $C_l \geq C-4\cdot \opt_2$.
\end{lemma}
\begin{proof}
	It suffices to consider the non-trivial case when $C \geq 4\cdot \opt_2$.
	Moreover, as we have $C_i\geq C$ for all machine $i\in R_{\geq k}$, it suffices to consider the machines $l$ in $R_{k-1}$.
	Let $J'$ be the jobs that contribute to the last $2\cdot \opt_2$ increase of makespan on machines in $R_{\geq k}$. 
	In other words, before the first job in $J'$ arrives, the makespan of each machine $i\in R_{\geq k}$ is at most $C_i - 2\cdot \opt_2$.
	
	Observe that the total size of the jobs in $J'$ is at least $2\cdot \opt_2 \cdot \sum_{i\in R_{\geq k}}s_i$, while the total size of jobs the optimal schedule processes on machines $R_{\geq k}$ is at most $\opt_2 \cdot \sum_{i\in R_{\geq k}}s_i$.	
	Hence we know that in $\opt_2$, at least one of the jobs in $J'$, say job $j$, is processed on some machine not in $R_{\geq k}$.	
	Since the processing time of job $j$ on machines in $R_{k-1}$ is at most $\opt_2$, and the machines not in $R_{\geq k}$ have speeds at most $2^{k-1}$, we have $p_j \leq \opt_2 \cdot 2^{k-1}$.
	
	Now consider the time when job $j$ arrives.
	
	By definition, in our mechanism job $j$ is assigned to some machine $i\in R_{\geq k}$.
	Moreover, we have
	\begin{equation*}
	\textstyle	C_i(j) + \frac{p_j}{s_i} \geq C_i - 2\cdot \opt_2 \geq C - 2\cdot \opt_2.
	\end{equation*}
	
	Let $i'\in R_{k-1}$ be the machine with minimum makespan before job $j$ arrives.
	If $C_{i'}(j) \geq C - 4\cdot \opt_2$ then we are done, as the makespans of machines do not decrease.
	
	Assume the contrary that $C_{i'}(j) < C - 4\cdot \opt_2$.
	Then we have
	\begin{align}
	C_i(j) - C_{i'}(j) & \textstyle \geq C - \left(2\cdot \opt_2+\frac{p_j}{s_i}\right) - (C - 4\cdot \opt_2) \nonumber \\
	& \textstyle = 2\cdot \opt_2 - \frac{p_j}{s_i}. \label{eq:lower-bound}
	\end{align}
	
	Recall that for any machine $l$, we have
	\begin{equation*}
	\pi_l = \frac{s_l}{s_{l+1}}\cdot (C_{l+1}(j) - C_l(j)) \leq \frac{1}{2}\cdot (C_{l+1}(j) - C_l(j)),
	\end{equation*}
	and the price on machine $l$ is $\sum_{l'\geq l} \pi_{l'}$.
	Hence the price on machine $i'$ is at most $\frac{1}{2}\left( C_i(t) - C_{i'}(t) \right)$ larger than that of machine $i$.
	
	Recall that the cost of going to machine $i$ is given by
	\begin{equation*}
	\textstyle	C_i(j) + \frac{p_j}{s_i} + \sum_{l\geq i}\pi_l.
	\end{equation*}
	
	On the other hand, the cost of going to machine $i'$ is at most
	\begin{align*}
	&\textstyle C_{i'}(j) + \frac{p_j}{2^{k-1}} + \frac{1}{2}\left( C_i(j) - C_{i'}(j) \right) + \sum_{l\geq i}\pi_l \\
	\leq &\textstyle C_{i}(j) + \opt_2 - \frac{1}{2}\left( C_i(j) - C_{i'}(j) \right) + \sum_{l\geq i}\pi_l \\
	\leq &\textstyle C_{i}(j) + \frac{1}{2}\cdot \frac{p_j}{s_i} + \sum_{l\geq i}\pi_l,
	\end{align*}
	where the last inequality holds by Inequality~\eqref{eq:lower-bound}.
	
	Thus we have a contradiction, as the cost is strictly smaller when job $j$ goes to machine $i'$ while the mechanism assigns job $j$ to machine $i$.
\end{proof}

\begin{lemma} \label{lemma:logm-competitive}
	The scheduling algorithm is $O(\log m)$-competitive.
\end{lemma}
\begin{proof}
	Recall that $\alg = C_m$ and $R_K$ contains the machines with highest speed.	
	For all $i\in R_K$, we have $C_i \geq \alg - \opt_2$, as no job has processing time strictly larger than $\opt_2$ on the fastest machines.	
	By applying Lemma~\ref{lemma:difference-of-C} recursively, all machines in $R_{\geq K-\log m}$ (which are those of speed at least $\frac{s_m}{m}$) have completion time at least $\alg - (4\log m+1)\cdot \opt_2$.
	
	If $\alg - (4\log m+1)\cdot \opt_2 \geq 2\cdot \opt_2$, then the total size of all jobs is at least $2\cdot \opt_2 \cdot \sum_{j\in R_{\geq K-\log m}}s_i$. On the other hand, as the optimal solution finishes all jobs within makespan $\opt_2$, the total size of jobs is upper bounded by
	\begin{align*}
	\opt_2 \cdot \sum_{j\in[m]}s_i & \textstyle\leq \opt_2 \cdot \left( m\cdot \frac{s_m}{m} + \sum_{j\in R_{\geq K-\log m}}s_i\right) \\
	&\textstyle  < 2\cdot \opt_2 \cdot \sum_{j\in R_{\geq K-\log m}}s_i.
	\end{align*}	
	Hence we have $\alg - (4\log m+1)\cdot \opt_2 < 2\cdot \opt_2$, which implies $\alg = O(\log m)\cdot \opt_2 = O(\log m)\cdot \opt$, as desired.
\end{proof}

In the following, we show that if the online jobs have bounded sizes, i.e., $p_j = O(1)$ for all job $j$, then the competitive ratio of mechanism ${\mathcal M}_\text{PPR}$ improves to $O(1)$.

Let $p_\text{max} = \max_{j\in J}\{p_j\}$ be the maximum size of an online job.

\begin{lemma} \label{lemma:O(1)-comp}
	The competitive ratio of the algorithm is $O(p_\text{max})$.
\end{lemma}
\begin{proof}
	We first show that the difference in makespan between two machines of similar speed are bounded (in terms of $p_\text{max}$).
	
	First, observe that for any $k\leq K$, the makespans of any two machines in $R_k$ differ by at most $\frac{p_\text{max}}{2^k}$.
	For any machine $i\in R_k$ and $i'\in R_{k+1}$, we have
	\begin{equation*}
	\textstyle	C_{i'} \leq C_{i} + \left(\frac{1}{s_i}+\frac{1}{s_{i'}} \right)\cdot p_\text{max} = C_{i} + \frac{3}{s_{i'}} \cdot p_\text{max},
	\end{equation*}
	as otherwise the last job that goes to machine $i'$ should have been assigned to machine $i$.
	
	As before, suppose that $\alg = C_m$.
	Let $s_m = 2^K$, and $i$ be the slowest machine that has non-zero workload.
	Suppose $s_i = 2^k$, then any machine $j$ with speed at least $2^k$ has makespan at least
	\begin{align*}
	& \textstyle C_j \geq C_i-\frac{p_\text{max}}{s_i} \geq C_m - p_\text{max}\cdot \sum_{l=k+1}^K \frac{3}{2^l} -\frac{p_\text{max}}{2^k} \\
	\geq &\textstyle \alg - \frac{3\cdot p_\text{max}}{2^k} -\frac{p_\text{max}}{2^k} = \alg - \frac{4\cdot p_\text{max}}{2^k}.
	\end{align*}
	
	Let $2^{k'}$ be the highest speed of a machine that is slower than $i$, where $k' \leq k-1$.
	Then we have
	\begin{equation*}
	\textstyle	\alg - \frac{4\cdot p_\text{max}}{2^k} \leq (\frac{1}{2^k}+\frac{1}{2^{k'}})\cdot p_\text{max}.
	\end{equation*}
	
	Now we consider the optimal schedules.
	
	If in the optimal solution any job is scheduled on a machine with speed slower than $2^k$, then we have $\opt \geq \frac{1}{2^{k'}}$, which implies
	\begin{equation*}
	\textstyle	\alg \leq (\frac{1}{2^k}+\frac{1}{2^{k'}})\cdot p_\text{max} + \frac{4\cdot p_\text{max}}{2^k} = O(p_\text{max})\cdot \opt .
	\end{equation*}
	
	Otherwise we know that the optimal solution also schedules jobs only on machines with speed at least $2^k$, which implies $\opt\geq \alg - \frac{4\cdot p_\text{max}}{2^k}$, as otherwise the optimal solution cannot finish all jobs.
	If the optimal solution schedules some job on a machine with speed $2^k$, then we have $\opt\geq \frac{1}{2^k}$;
	otherwise we have $\opt \geq C_i \geq \frac{1}{2^k}$ (since all machines of speed at least $2^{k+1}$ has makespan at least $C_i$).
	In both cases, we have $\alg \leq \opt + \frac{4\cdot p_\text{max}}{2^k} = O(p_\text{max})\cdot \opt$.
\end{proof}

Finally, by the design of mechanism ${\mathcal M}_\text{PPR}$, it is easy to see that ${\mathcal M}_\text{PPR}$ satisfies the following properties.
Let $L_{i}$ be the set of jobs assigned to machine $i$ at the end of the mechanism.
A mechanism is {\em fair} if for any two machines $i$ and $i'$ such that $s_{i}=s_{i'}$, 
$\sum_{j\in L_{i}}p_{j} \geq \sum_{l\in L_{i'}\backslash\{k\}}p_{l}$ where $k\in L_{i'}$ is the last job assigned to $i'$;
is {\em anonymous} if the scheduling depends only on the configuration $\pi = (s_1,\ldots,s_m)$ of machine speeds.
In other words, if a machine submits speed $s_i$ while the configuration is given by $\pi$, then the jobs assigned to the machine is given by $L_{i}$, where $L_{i}$ is the set of jobs assigned to machine $i$ according to permutation $\pi$.

\begin{lemma}
	Mechanism ${\mathcal M}_\text{PPR}$ is fair and anonymous.
\end{lemma}

\section{Truthfulness for Machines} \label{sec:truthful-cases}

Although well-behavior is important on its own behalf, as it guarantees fairness, anonymity and symmetry (see the paragraph above and Section~\ref{ssec:our_results}) between machines,
in this section, we show that 
well-behavior is indeed useful towards achieving truthfulness against selfish machines.
We show several cases for which our mechanism is truthful for both selfish jobs and machines.
Recall that the machines are oblivious of the online jobs.
Given that there are many cases under which the mechanism is monotone, selfish machines are more likely to be truthful, so as to get rid of the risk of a decrease in the utility.

Recall that to achieve truthfulness for machines, it suffices to show that the mechanism is monotone~\cite{stacs/AndelmanAS05}.
First, we show that when the online jobs have unit size, then mechanism ${\mathcal M}_\text{PPR}$ is monotone.

\begin{lemma}\label{lem:truthful:unit}
	Mechanism ${\mathcal M}_\text{PPR}$ is truthful when jobs have unit size.
\end{lemma}
\begin{proof}
	Suppose $s_i = 2^k$ for some machine $i$.
	Note that as long as $s_i \in [ 2^k, 2^{k+1} )$, the schedule remains the same.	
	Hence it suffices to show that if the claimed speed of machine $i$ increases from $2^k$ to $2^{k+1}$, the number of jobs assigned to it does not decrease.
	Assume for the sack of contradiction that machine $i$ is assigned $L$ jobs with speed $2^k$ but assigned less than $L$ jobs with speed $2^{k+1}$.
	Moreover, we can assume w.l.o.g. that the $L$-th job assigned to machine $i$ when $s_i = 2^k$ is the last job of the whole sequence (otherwise we can remove the later jobs to achieve a smaller counter-example).
	
	In the following, we order the machines by their speeds: $s_1 \leq s_2\leq \ldots\leq s_m$.
	
	First, observe that for any job $j$ and any machine $l < m$, we have $C_{l+1}(j) \leq C_l(j)+\frac{1}{s_l}$, as otherwise the last job that goes to machine $l+1$ should have been assigned to machine $l$.
	
	Since the makespan of machine $i$ when $s_i = 2^{k+1}$ is at most $\frac{L-1}{2^{k+1}}$, the workload on any machine in $R_{\geq k+1}$ is at most
	\begin{equation*}
	\textstyle \frac{L-1}{2^{k+1}} + \sum_{l\geq k+1}\frac{1}{2^l} < \frac{L+1}{2^{k+1}} \leq \frac{L}{2^k}.
	\end{equation*}
	
	In other words, the workload on any machine in $R_{\geq k+1}$ decreases after machine $i$ claims a higher speed.
	
	In the following, we show that the number of jobs assigned to any machine with speed at most $2^k$ does not increase.
	
	First, for any machine in $R_k$, its makespan when $s_i = 2^k$ is at least $\frac{L-1}{2^k}$, which is larger than $\frac{L-1}{2^{k+1}}$, the upper bound on its makespan when $i$ claims speed $2^{k+1}$.
	It remains to consider machines with speed at most $2^{k-1}$.
	Fix any machine $j$ with speed $s_j \leq 2^{k-1}$.
	We show that its workload does not increase.
	
	Consider the moment when the $L$-th job goes to machine $i$ when $s_i = 2^k$.
	Since the job goes to machine $i$ instead of $j$, we have
	\begin{equation*}
	\textstyle	C_j + \frac{1}{2}(\frac{L-1}{2^k} - C_j) + \frac{1}{s_j} \geq \frac{L}{2^k},
	\end{equation*}
	where $\frac{1}{2}(\frac{L-1}{2^k} - C_j)$ is an upper bound on the difference between the prices of machine $i$ and $j$.
	Hence we have
	\begin{equation*}
	\textstyle	C_j \geq \frac{L}{2^{k-1}} - \frac{L-1}{2^k} - \frac{2}{s_j} =  \frac{L+1}{2^{k}} - \frac{2}{s_j}.
	\end{equation*}
	
	In other words, when $s_i = 2^k$, the number of jobs on machine $j$ is at least $\lceil \frac{(L+1)s_j}{2^k} \rceil  - 2$.
	On the other hand, when $s_i = 2^{k+1}$, the number of jobs on machine $j$ is at most $\lfloor \frac{(L-1)s_j}{2^{k+1}} \rfloor$.
	Observe that
	\begin{equation*}
	\textstyle \left\lceil\frac{(L+1)s_j}{2^k}\right\rceil - \left\lfloor\frac{(L-1)s_j}{2^{k+1}}\right\rfloor \geq \left\lceil \frac{(L+3)s_j}{2^{k+1}} \right\rceil \geq \left\lfloor \frac{(L-1)s_j}{2^{k+1}} \right\rfloor + 1.
	\end{equation*}
	
	If $\lfloor \frac{(L-1)s_j}{2^{k+1}} \rfloor = 0$, then machine $j$ is not assigned any job when $s_i = 2^{k+1}$ (hence its workload does not increase); otherwise
	\begin{equation*}
	\textstyle \left\lceil \frac{(L+1)s_j}{2^k} \right\rceil  - 2 \geq \left\lfloor \frac{(L-1)s_j}{2^{k+1}} \right\rfloor.
	\end{equation*}
	
	In both cases, the number of jobs assigned to machine $j$ does not increase when machine $i$ claims a higher speed.
\end{proof}

Finally, we show that our mechanism is truthful for the case $m=2$.
However, instead of rounding the machine speeds to powers of $2$, in the following proof, we assume that the speeds of machines are rounded down to powers of $4$.

\begin{lemma}
	Mechanism ${\mathcal M}_\text{PPR}$ is truthful when $m=2$.
\end{lemma}
\begin{proof}
	As before, we prove that the mechanism is monotone.
	Recall that by scaling we can assume that the slower machine has speed $1$.
	It suffices to show that if the speed of the slower machine is increased, then its load does not decrease.
	The case when we increase the speed of the faster machine is equivalent to the case when we decrease the speed of the slower machine.
	
	The case when the machines have the same speed $1$ is relatively easier.
	Let $L\leq R$ be the makespans (workloads) on the two machines.
	Suppose we increase the speed of one of the machine to $4$, then its workload is at least $\frac{4}{5}(L+R)$, which is at least $R$ if $R\leq 4 L$.
	If $R>4 L$, then the last job, which is of size at least $R-L > 3 L$, must be assigned to the faster machine. Thus the workload of the faster machine is at least $(R-L) +\frac{4}{5}\cdot 2L \geq R$. 
	
	Now suppose that the speed of machine $1$ is $1$, while the speed of machine $2$ is $4^i$, where $i\geq 1$. 
	We prove that for every job assigned to machine $1$ when $s_1 = 1$, if it is not assigned to machine $1$ when $s_1 = 4$, then the workload on the machine is no smaller than its original workload after the assignment, which implies that the workload does not decrease.
	
	Let $P$ be the size of the job and consider the moment when the job arrives, when $s_1 = 1$.
	Let $L$ and $R$ be the makespans of the two machines, where $L\leq R$.
	Since the job is assigned to machine $1$, we have $P\leq R-L$.
	
	Let $L'$ and $R'$ be defined similarly, when $s_1 = 4$.
	Since $P$ is not assigned to machine $1$ after the increase, we have $L'+\frac{P}{4} \geq R'$.
	Then we have
	\begin{align*}
	4\cdot L' &\textstyle \geq 4\cdot R' - P \geq 4\cdot(R-\frac{4\cdot L' - L}{4^i}) - (R-L)\\
	&\textstyle \geq 4\cdot \frac{R}{2} - R = R \geq L+P,
	\end{align*}
	where in the second inequality we use $L+4^i\cdot R = 4\cdot L' + 4^i\cdot R'$.
	Hence at any moment, the workload of machine $1$ when $s_1 = 4$ is no smaller than its workload when $s_1 = 1$.
\end{proof}

\section{Conclusion}

In this paper, we study well-behaved mechanisms for online strategic jobs.
We first present two mechanisms that are truthful for selfish jobs and machines, with competitive ratios $m$ and $\frac{2 s_\text{max}}{s_\text{min}}$, respectively.
Then we consider well-behaved schedules and show that any well-behaved mechanism must have a competitive ratio at least $\Omega(\sqrt{m})$.

We propose an almost well-behaved dynamic posted-price mechanism that is truthful for the jobs and has a competitive ratio of $O(\log m)$.
When the sizes of all jobs are bounded, the competitive ratio of our mechanism improves to be constant. 	
Finally, we show that when all jobs have unit size or there are only two machines,
our mechanism is actually truthful against selfish machines. 

We expect our algorithm and analysis to shed some lights on the design of competitive mechanisms that are truthful for both the selfish jobs and selfish machines.
It is also interesting to improve the ratio of $O(\log m)$ for almost well-behaved mechanisms.


\bibliographystyle{ACM-Reference-Format}  
\balance  
\bibliography{truthful}  

\end{document}